\documentclass[12pt,letterpaper]{article}
\pdfoutput=1

\usepackage[svgnames]{xcolor}
\usepackage{amsmath}
\usepackage{amsthm}
\usepackage{amssymb}

\usepackage{enumitem}

\usepackage{hyphenat}
\usepackage{algorithm}
\usepackage{algpseudocode}
\usepackage{algorithmicx}
\usepackage[round, sectionbib]{natbib}

\newtheorem{theorem}{Theorem}[section]

\newtheorem{lemma}[theorem]{Lemma}

\newtheorem{defn}[theorem]{Definition}

\newtheorem{example}[theorem]{Example}
\newtheorem{problem}[theorem]{Problem}

\newcommand{\norm}[1]{\|#1\|}
\newcommand{\D}{\partial}

\newcommand{\RR}{{\mathbb{R}}}

\newcommand{\degD}{{\deg_{\D}}}

\newcommand{\ftil}{{\widetilde f}}
\newcommand{\gtil}{{\widetilde g}}
\newcommand{\fhat}{{\widehat f}}
\newcommand{\ghat}{{\widehat g}}
\newcommand{\calC}{{\mathcal C}}
\newcommand{\Vhat}{{\widehat V}}
\newcommand{\What}{{\widehat W}}

\newcommand{\what}{{\widehat w}}

\newcommand{\Sigmabar}{{\overline\Sigma}}

\newcommand{\lclm}{{\mbox{\upshape lclm}}}

\newcommand{\gcrd}{{\mbox{\upshape gcrd}}}

\renewcommand{\epsilon}{\varepsilon}


\DeclareMathOperator{\diag}{{\mbox{\upshape diag}}}

\DeclareMathOperator{\lnullspace}{{\mbox{\upshape null}_{\ell}}}

\DeclareMathOperator{\lcoeff}{{\mbox{\upshape lcoeff}}}
\DeclareMathOperator{\content}{{\mbox{\upshape cont}}}

\definecolor{darkgreen}{rgb}{0,.35,0}
\definecolor{darkblue}{rgb}{0,0,.5}
\definecolor{darkred}{rgb}{.6,0,0}

\newcommand{\papertitle}{{Computing GCRDs of Approximate Differential Polynomials}}

\usepackage[pdfauthor={Mark Giesbrecht, Joseph Haraldson},
            bookmarksnumbered=true,colorlinks,linkcolor=darkblue,
            citecolor=darkgreen, urlcolor=darkred,
            pdftitle=\papertitle]{hyperref}

\begin{document}

\title{\papertitle}

 \author{
 Mark Giesbrecht and Joseph Haraldson\\[2pt]
        {Cheriton School of Computer Science}\\
        {University of Waterloo}\\
        {Waterloo, Ontario}\\
        \href{mailto:mwg@uwaterloo.ca}{\{mwg,jharalds\}@uwaterloo.ca}}
\maketitle

\begin{abstract} 
  Differential (Ore) type polynomials with approximate polynomial	
  coefficients are introduced.  These provide a useful representation
  of approximate differential operators with a strong algebraic
  structure, which has been used successfully in the exact, symbolic,
  setting.  We then present an algorithm for the approximate Greatest
  Common Right Divisor (GCRD) of two approximate differential
  polynomials, which intuitively is the differential operator whose
  solutions are those common to the two inputs operators.  More
  formally, given approximate differential polynomials $f$ and $g$, we
  show how to find ``nearby'' polynomials $\ftil$ and $\gtil$ which
  have a non-trivial GCRD.  Here ``nearby'' is under a suitably
  defined norm.  The algorithm is a generalization of the SVD-based
  method of \cite{CGTW95} for the approximate GCD of regular
  polynomials.  We work on an appropriately ``linearized''
  differential Sylvester matrix, to which we apply a block SVD.  The
  algorithm has been implemented in Maple and a demonstration of its
  robustness is presented. 
\end{abstract}




 \section{Introduction}
\label{sec:intro}

The ring of differential (Ore) polynomials $\RR(t)[\D;']$ over the
real numbers $\RR$ provides a (non-commutative) polynomial ring
structure to the linear ordinary differential operators.  Differential
polynomials have found great utility in symbolic computation, as they
allow us to apply algebraic tools to the simplification and solution
of linear differential equations; see \cite{BroPet94} for a nice
introduction to the mathematical and computational aspects.  The ring of
differential polynomials $\RR(t)[\D;']$ is defined as the usual
polynomials in $\RR(t)[\D]$ (i.e., polynomials in $\D$ over the
commutative field of rational functions $\RR(t)$), under the usual
polynomial addition and the non-commutative multiplication rule
\[
\D f(t) = f(t) \D + f'(t) \text{ for } f(t) \in \RR(t),
\]
where $f'(t)$ is the usual formal derivative of $f(t)$ with respect to
$t$. This definition of $\RR(t)[\D;']$ is useful because there is a
natural action of $\RR(t)[\D;']$ on the space $\calC^\infty$ of
infinitely differentiable functions $y(t):\RR\to\RR$.  In particular,
for any $y(t)\in\calC^\infty$,
\[
f(\D) = \sum_{0\leq i\leq k} f_i(t)\D^i
~~ \mbox{acts on} ~ y(t) ~ \mbox{as} ~~
\sum_{0\leq i\leq k} f_i(t) \frac{d^i}{dt^i} y(t).
\]

A primary benefit of viewing differential operators in this way is
that they have the structure of a left (and right) Euclidean domain.
In particular, for any two polynomials $f,g \in \RR(t)[\D;']$, there is a
unique polynomial $h\in\RR(t)[\D;']$ of maximal degree in $\D$ such
that $f=uh$ and $g=vh$ for $u,v\in\RR(t)[\D;']$ (i.e., $h$ divides $f$
and $g$ exactly on the right).  This polynomial $h$ is called the GCRD of 
$f$ and $g$ and it is unique up to a scalar
multiplication by a non-zero element of $\RR(t)$ (we could make this
GCRD have leading coefficient 1, but that would introduce denominators from $\RR[t]$, as
well as potential numerical instability, as we shall see).

The important geometric interpretation is that there is an algorithm
to determine the differential polynomial whose solution spaces is the
intersection of the solution space of $f$ and $g$; this is precisely
$h=\gcrd(f,g)$.

The goal of this paper is to devise an efficient, numerically robust
algorithm to compute the GCRD when the coefficients in $\RR$ are given
approximately.  Specifically, given $f,g \in \RR(t)[\D;']$, we wish to
find $\ftil,\gtil\in\RR(t)[\D;']$, where $\ftil$ is ``near'' $f$ and
$\gtil$ is ``near'' $g$ such that $\degD \gcrd(\ftil,\gtil)\geq
1$. That is, $\ftil$ and $\gtil$ have a \emph{non-trivial} GCRD.  The
precise definition of nearness is given below.  The approach we will
take is similar to the one developed in \cite{CGTW95} for regular polynomials,
which is to reduce the problem to a singular value decomposition or
total least squares problem.

The problem of computing the GCRD in a symbolic and exact setting dates
back to \cite{Ore33}, who presents a Euclidean like algorithm.  See
\cite{BroPet94} for an elaboration of this approach.  \cite{Li97}
introduce a differential-resultant-based algorithm which makes
computation of the GCRD very efficient using modular arithmetic.  We
generalize and adapt their approach to a numerical setting here.

The analogous approximate GCD problem for approximate regular
(commutative) polynomials has been a key topic of research in
symbolic-numeric computing since its inception.  A full survey is not
possible here, but we note the deep connection between our current
work and that of \cite{CGTW95}; see also \cite{KarLak96},
\cite{SasSas97}, and \cite{ZenDay04}.  Also important to this current work
is the use of so-called structured (numerical) matrix methods for
approximate GCD, such as structured total least squares (STLS) and
structured total least norm (STLN); see \cite{BGM05} and \cite{KYZ05}.
A structured approach to relative primality is taken in
\cite{Beckerman97}.
More directly employed later in this paper is the multiple polynomial
approximate GCD method of \cite{KalYan06}.
This latter paper also provides a nice survey of recent developments.

\subsection{Differential polynomial basics}
\label{ssec:prelim}

The ring $\RR(t)[\D;']$ is a non-commutative principal left (and
right) ideal domain.  For $f,g \in R(t)[\D;']$, with $\degD f=n$ and
$\degD g=m$, we have the following properties \citep{Ore33}.

\begin{enumerate}
\item $\degD (fg) = \degD f + \degD g$ , \newline
  $\degD(f+g) \leq \max \{ \degD  f, \degD g\}$.
\item There exist $q,r\in \RR(t)[\D;']$ with $\degD r < \degD g$ such
  that $f=qg+r$ (Division with Remainder).
\item There exists $h\in \RR(t)[\D,']$ of maximal degree in $\D$ with
  $f=w_1h$ and $g=w_2h$. $h$ is called the GCRD (Greatest Common Right
  Divisor) of $f$ and $g$.
\item There exist $w_3,w_4 \in R(t)[\D,']$ such that $w_3 f= w_4 g =
  h$ for $h$ of minimal degree. $h$ is called the LCLM (Least Common
  Left Multiple) of $f$ and $g$.
\item $\degD \lclm(f,g) = \degD f + \degD g - \degD \gcrd(f,g)$.
\end{enumerate}

These immediately imply the following characterization of a
non-trivial GCRD.

\begin{lemma}
  \label{lem:nontrivgcrd}
  Suppose $f,g \in \RR(t)[\D;']$, with $\degD f=m$ and $\degD g=n$.
  Then $\degD \gcrd(f,g) \geq 1 $ if and only if there exists $u,v \in
  \RR(t)[\D;']$ such that $\degD u< n$, $\degD v< m$, and
  $uf+vg = 0$.
\end{lemma}

This lemma will enable us to set up a resultant-like linear system
over $\RR$ for the GCRD, which will lead to the desired algorithms.

Because of the non-commutative property of $\RR(t)[\D;']$, it will be
important to maintain a canonical form for any $f\in\RR(t)[\D;']$.  We
will always write
\[
f=\frac{1}{f_{-1}(t)} \sum_{i=0}^m f_i(t) \D^i,
\]
for polynomials $f_{-1},f_0,\ldots,f_m \in \RR[t]$, with coefficients
in $\RR(t)$ always written to the left of powers of $\D$.  Moreover,
for $f$ as above,  and $\ell> \deg f$, we define
\[
\Psi_{\ell}(f) = \frac{1}{f_{-1}}\cdot (f_0,
f_1,\ldots,f_m,0,\ldots,0)\in\RR(t)^\ell.
\]
I.e., $\Psi_{\ell}$ maps polynomials in $\RR(t)[\D;']$ of degree (in
$\D$) less than $\ell$ into $\RR(t)^\ell$.

It will also be useful to ensure that our differential polynomials are
\emph{primitive}.  

\begin{defn}[Primitive Differential Polynomial]
  Let $f \in \RR[t][\D;']$ where $\degD f = m$ is in standard
  form. The content of $f$ is given by $\content(f) =
  \gcd(f_0,f_1,\ldots,f_m)$.  If $\content(f) = 1$, we say that $f$ is
  primitive.
\end{defn}

For our primary problem of computing GCRDs of differential polynomials
$f,g\in\RR(t)[\D;']$, we will assume both that the coefficients of
$f,g$ are polynomials in $\RR[t]$, and that $f$ and $g$ are primitive.
In the exact setting this is clearly without loss of generality, since
non-zero elements of $\RR(t)$ are units (and hence we can multiply and
divide on the left by them).  For approximate differential polynomials we
must compute the approximate GCD of a number of polynomials in
$\RR[t]$.  This is in itself an important research problem, but has
been considered deeply in \cite{KalYan06}, which also contains a
useful and current survey of related approximate GCD results.

\subsection{Norms of differential polynomials}
\label{ssec:norm}

To provide our notion of approximate differential polynomials a formal
meaning we need a proper definition of the \emph{norm} of a
differential polynomial. For this we will use the \emph{coefficient
  2-norm} as follows.
\begin{defn}~\\[-\baselineskip]
  \begin{itemize}
  \item[(i)] For (a regular polynomial) $p=\sum_{0\leq i\leq k} p_it^i\in\RR[t]$, define
   $\norm{p}=\norm{p}_2 = (\sum_{0\leq i\leq k} p_i^2)^{1/2}$.
 \item[(ii)] For $f=f_0+f_1\D+\cdots+f_m\D^m\in \RR[t][\D;']$ define \linebreak
   $\norm{f}=\norm{f}_2 = \left (\sum_{0\leq i\leq m} \norm{f_i}^2
   \right ) ^{1/2} $.
\end{itemize}
\end{defn}

Note that we are assuming that our coefficients are polynomials in
$\RR[t]$, and not rational functions.  One could extend the definition
of norm to encompass coefficients in $\RR(t)$, but it will not be
necessary in this paper.

\subsection{The approximate GCRD problem}

We can now formally state the main problem under consideration in this
paper.

\begin{problem}\label{prb:main-prob}
  Given $f,g\in\RR[t][\D;']$, find a small $\epsilon>0$ and
  $\ftil,\gtil\in\RR[t][\D;']$ with $\norm{f-\ftil}<\epsilon\norm{f}$
  and $\norm{g-\gtil}<\epsilon\norm{g}$ such that
  $\degD\gcrd(\ftil,\gtil)\geq 1$.
\end{problem}

That is, we are looking for ``nearby'' differential polynomials, in
the coefficient 2-norm, which possess a non-trivial GCRD.

 
%
 \section{GCRD via Linear Algebra}
\label{sec:lin-alg}

In this section we demonstrate how to reduce the computation of the
GCRD to that of linear algebra over $\RR(t)$, and then over $\RR$
itself.  This approach has been used in the exact computation of GCRDs
(see \citep{Li97}) and Hermite forms \citep{GieKim13}, and has the
benefit of reducing differential, and more general Ore problems, to a
system of equations over a commutative field.  Here we will show that
it makes our approximate version of the GCRD problem amenable to
numerical techniques.

\subsection{Reduction to linear algebra over {\large $\RR(t)$}}
\label{ssec:dsyl}

Let $f,g \in \RR(t) [\D;']$ have degrees in $\D$ of $m$ and $n$
respectively. Then by Lemma \ref{lem:nontrivgcrd}, $\degD \gcrd(f,g)
\geq 1$ if and only if there exists $u,v \in \RR(t)[\D;']$ such that
$\degD u < n$ and $\degD v < m$ and $uf+vg = 0$. We can encode the
existence of $u,v$ as an $(m+n)\times (m+n)$ matrix over $\RR(t)$ as
follows.  For convenience define the matrix 
\[
V = V(f,g) = 
\begin{pmatrix}
  \Psi_{m+n}(f) \\
  \Psi_{m+n}(\D f) \\
  \vdots \\
  \Psi_{m+n}(\D^{n-1} f) \\
  \Psi_{m+n}(g) \\
  \Psi_{m+n}(\D g) \\
  \vdots \\
  \Psi_{m+n}(\D^{m-1} g)
\end{pmatrix} \in \RR(t)^{(m+n) \times (m+n)},
\]
the differential Sylvester matrix of $f$ and $g$  (see \citep{Li97}),
analogous to the Sylvester matrix for usual polynomials (see, e.g.,
\cite[Chapter 6]{MCA3}).   

The utility of this comes in the following observation.  Let
\[
u = \sum_{0\leq i\leq n-1} u_i \D^i, ~~
v = \sum_{0\leq i\leq m-1} v_i \D_i \in\RR(t)[\D;'],
\]
and
\[
w = (u_0,u_1,\ldots, u_{n-1}, v_0, v_1, \ldots, v_{m-1}) \in \RR(t)^{1
  \times (m+n)}.
\] 
Then $uf+vg=0$ implies $wV = 0$.  This means that $w$ is a non-trivial
vector in the nullspace of $V$, and in particular, $V$ is singular.
Clearing denominators of $f$ and $g$ we may assume that $u,v \in
\RR[t][\D;']$, i.e., they have polynomial coefficients, which implies
that $V \in \RR[t]^{(m+n) \times (m+n)}$.  Moreover, if $f,g \in
\RR[t][\D;']$ have degrees in $t$ at most $d$ then $\deg_t V_{ij}\leq
d$.

\begin{lemma}
  \label{lem:RRtGCRD}
  Suppose $f,g\in\RR[t][\D;']$, where $\degD f=m$, $\degD g=n$ and $\deg_t
  f\leq d$ and $\deg_t g\leq d$.
  \begin{itemize}
  \item $V=V(f,g)$ is singular if and only $ \degD \gcrd(f,g)\geq 1$.
  \item $\degD\gcrd(f,g)= \dim \lnullspace(V)$, where $\lnullspace(V)$
    is the left nullspace of $V$.
  \item For any 
    $w = (u_0,\ldots, u_{n-1}, v_0, \ldots, v_{m-1}) \in \RR(t)^{1
      \times (m+n)}$
    such that $wV=0$, we have $uf+vg=0$, where
    $u=\sum_{0\leq i<n} u_i\D^i$ and $v=\sum_{0\leq i<m} v_i\D^i$.
  \item Suppose that $\degD\gcrd(f,g)\geq 1$.  Then there exists a $w\in
    \RR[t]^{1\times (m+n)}$ such that $wV = 0$  and $\deg_t w \leq
    \mu = 2(m+n)d$.
  \end{itemize}
\end{lemma}
\begin{proof}
  Part (i) -- (iii) follow from Lemma \ref{lem:nontrivgcrd} and the
  discussion above.  Part (iv) follows from an application of
  Cramer's rule, and a bound on the degree of the determinants of a
  polynomial matrix.
\end{proof}

\begin{example}
 Let \[f ={{\it \D}}^{2}+ \left(  0.5\,t+ 1.0 \right) {\it \D}+ 0.3\,t+ 0.06\,{t
}^{2}+ 0.2\] and 
 \[g={{\it \D}}^{2}+ \left(  0.9\,{t}^{2}+ 1.0+ 0.2\,t \right) {\it \D}+
 0.2+ 0.9\,{t}^{2}+ 0.18\,{t}^{3}
.\] 
The corresponding differential Sylvester matrix $V$ is given by

\[
 \begin {pmatrix}  0.3\,t+ 0.06\,{t}^{2}+ 0.2& 0.5\,t+ 1.0&
1&0\\   0.3+ 0.12\,t& 0.7+ 0.06\,{t}^{2}+ 0.3\,t& 0.5
\,t+ 1.0&1\\   0.2+ 0.9\,{t}^{2}+ 0.18\,{t}^{3}& 0.9
\,{t}^{2}+ 1.0+ 0.2\,t&1&0\\   1.8\,t+ 0.54\,{t}^{2}&
 0.9\,{t}^{2}+ 0.18\,{t}^{3}+ 1.8\,t+ 0.4& 0.9\,{t}^{2}+ 1.0+ 0.2\,t&1
\end {pmatrix}. 
\]

$V$ has rank $3$ with the (left) null space vector 

\[
\left( \begin {array}{c} - 27.0\,{t}^{4}+ 9.0\,{t}^{3}+ 60.0\,t- 10.0
\\  - 30.0\,{t}^{2}+ 10.0\,t\\   9.0
\,{t}^{3}- 3.0\,{t}^{2}- 60.0\,t+ 10.0\\   30.0\,{t}^
{2}- 10.0\,t\end {array} \right)  ^T .
\]
\end{example}

\begin{defn}
For any matrix $V\in\RR(t)[\D;']$, we define the Frobenius norm
$\norm{V}_F$ by
\[
\norm{V}_F^2 = \sum_{ij} \norm{V_{ij}}^2.
\]
\end{defn}
\begin{lemma}
  \label{lem:vandnorm}
  Let $f,g\in\RR[t][\D;']$ have $\degD f\leq m$, $\degD g\leq n$, and
  both have degree in $t$ at most $d$. Let $V=V(f,g)$ be the
  differential Sylvester matrix of $f$ and $g$.  Then
  $\norm{V}_F^2\leq d^{2n}\norm{f}^2+d^{2m}\norm{g}^2$.
\end{lemma}
\begin{proof}
  First note that $\norm{\D f}^2\leq (d^2+1)\norm{f}^2$, and hence
  $\norm{\D^kf}^2\leq (d^2+1)^k\norm{f}^2$. Thus
  \begin{align*}
    \norm{V}_F^2 = & \sum_{0\leq i<n} \norm{\D^if}^2 + \sum_{0\leq i<m}
    \norm{\D^i g}^2\\
    \leq & \sum_{0\leq i<n} (d^2+1)^i\norm{f}^2 + \sum_{0\leq i<m} (d^2+1)^i\norm{g}^2\\
    \leq & d^{2n}\norm{f}^2 + d^{2m}\norm{g}^2. \qed
  \end{align*}
\end{proof}

Note that the exponentials of $d$ (or, more precisely, the falling
factorials) are intrinsic in the resultant formulation, but will cause
considerable numerical instability for large degrees in $\D$.  As is
typical with differential polynomials we generally restrict ourselves
to small degrees in $\D$.

\subsection{Reduction to linear algebra over {\large $\RR$}}
\label{ssec:idsyl}

Next we show how to encode the existence of a GCRD as a linear algebra
problem over $\RR$, as opposed to $\RR(t)$.  Again let $V\in \RR[t]^{
  (m+n) \times (m+n)}$ be the differential Sylvester matrix of $f,
g\in\RR[t][\D;']$ of degrees $m$ and $n$ respectively in $\D$, and
degrees at most $d$ in $t$.  From Lemma \ref{lem:RRtGCRD} we know that
if a GCRD exists then there is a $w\in\RR[t]^{1\times (m+n)}$ such
that $wV=0$, with $\deg_t w\leq \mu = 2(m+n)d$.

Now suppose $b = b_0 + b_1 t+\cdots + b_{\mu+d}t^{\mu+d} \in R[t]$ and
define 
\[
\Psi(b) =
(b_0,b_1,\ldots, b_{\mu+d}) \in \RR^{1 \times (\mu+d+1)}.
\] 
For any polynomial $a\in\RR[t]$ of degree at most $d$ let
\[
\Gamma(a) = \begin{pmatrix}
  \Psi(a) \\
  \Psi(t a) \\
  \vdots \\
  \Psi( t^{\mu} a)
\end{pmatrix} \in \RR^{(\mu+1) \times (\mu+d+1)}.
\]
$\Gamma(a)$ is the left multiplier matrix of $a$ with respect to the
basis $\langle 1,t,\ldots,t^{\mu+d} \rangle$.  

Now given the $(m+n) \times (m+n)$ matrix $V$ we apply $\Gamma$
entry-wise to $V$ to obtain $\Vhat\in \RR^{(m+n)(\mu+1) \times
  (m+n)(\mu +d +1)}$; each entry of $V$ in $\RR[t]$ is mapped to a
block entry $\RR^{(\mu+1)\times (\mu+d+1)}$ in $\Vhat$.  We refer to
$\Vhat$ as the \emph{inflated differential Sylvester matrix} of $f$
and $g$.

\begin{lemma}
  \label{lem:RRGCRD}
  Let $f,g\in\RR[t][\D;']$ be as above, with differential Sylvester
  matrix $V\in\RR[t]^{(m+n)\times (m+n)}$ and inflated differential
  Sylvester matrix $\Vhat\in\RR^{(m+n)(\mu+1)\times (m+n)(\mu+d+1)}$.
  There exists a $w\in \RR[t]^{1\times (m+n)}$ such that $wV=0$, if
  and only if there exists a $\what \in \RR^{(\mu+d+1)\times (m+n)(\mu+1)}$
  such that $\what \Vhat =0$.  More generally, 
  \[
  \degD\gcrd(f,g)= \frac{\dim\lnullspace(\Vhat)}{\mu+d+1}.
  \]
\end{lemma}
\begin{proof}
  This follows directly from the definition of $\Gamma$ and Lemma
  \ref{lem:RRtGCRD}.
\end{proof}

\begin{example}
  Consider $f = \left( 0.84\,t+ 0.45 \right) {\it \D} + 0.11\,t+ 0.42$
  and $g= 0.66\,{\it \D}+ 0.92\,t$. Then the matrix $\widehat V(f,g)$
  is given by 
  {\small
  \[
  \begin{pmatrix} 
    0.42& 0.11&0&0&0&0& 0.45& 0.84&0 &0&0&0 \\[2pt]
    0& 0.42& 0.11&0&0&0&0& 0.45& 0.84&0&0&0 \\[2pt]
    0&0& 0.42& 0.11&0&0&0&0& 0.45& 0.84&0&0 \\[2pt] 
    0&0&0& 0.42& 0.11&0&0&0&0& 0.45& 0.84&0 \\[2pt]
    0&0&0&0& 0.42& 0.11&0&0&0&0& 0.45& 0.84 \\[2pt] 
    0& 0.92&0&0&0&0& 0.66&0&0&0&0&0 \\[2pt] 
    0&0& 0.92&0&0&0&0& 0.66&0&0&0&0 \\[2pt] 
    0&0&0&   0.92&0&0&0&0& 0.66&0&0&0 \\[2pt] 
    0&0&0&0& 0.92&0&0&0&0& 0.66&0&0 \\[2pt] 
    0&0&0&0&0& 0.92&0&0&0&0& 0.66&0
  \end{pmatrix}
  \]}
\end{example}

We can now bound the norm of the inflated differential Sylvester
matrix.

\begin{lemma}
  \label{lem:infvandnorm}
  Let $f,g\in\RR[t][\D;']$ have $\degD f\leq m$, $\degD g\leq n$ and
  both have degree degree at most $d$ in $t$.  Let
  $\Vhat\in\RR^{(m+n)(\mu+1)\times (m+n)(\mu+d+1)}$ be the inflated
  differential Sylvester matrix of $f$ and $g$, where $\mu=2(m+n)d$.
  Then $\norm{\Vhat}_2\leq \mu \cdot (d^{2n}\norm{f}^2+d^{2m}\norm{g}^2)$.
\end{lemma}
\begin{proof}
  Each row of $\Vhat$ consists precisely of entries of $V=V(f,g)$,
  shifted in position with respect to the previous row.  Thus
  \begin{align*}
  & \norm{\Vhat}_F^2 \leq \mu\cdot \norm{V}_F^2 \leq \mu \cdot
  (d^{2n}\norm{f}^2+d^{2m}\norm{g}^2), ~\mbox{and}\\
  & \norm{\Vhat}_2^2 \leq \norm{\Vhat}_F^2. 
  \end{align*}
  See \cite[\S2.3.2]{GolLoa13}.
\end{proof}


 \section{Computing an approximate GCRD}
\label{sec:approxgcrd}

We have now formulated the problem of determining the existence of
GCRD's of differential polynomials in $\RR[t][\D;']$ as one of
computing left null vectors of the inflated differential Sylvester matrix over
$\RR$.  We can now adapt the approach of \cite{CGTW95} of using the
SVD to find the nearest singular matrix.  While this will not be
perfect, in that the nearest singular matrix will not generally have
the same structure as the inflated differential Sylvester matrix, if our input
differential polynomials are ``nearby'' polynomials with a non-trivial
GCRD we will generally recover them.  

For convenience we will generally assume throughout this section that
our input differential polynomials are normalized, that is have
coefficient 2-norm $1$ under the definition of Section
\ref{ssec:norm}.  This can, of course, be enforced by a simple a
priori renormalization, i.e., dividing through by the actual norm, and
does not affect the generality or quality of the results.

\subsection{Finding nearby non-trivial GCRDs}
\label{sec:lsquare}

It is well understood how to find the nearest singular
\emph{unstructured} matrix to a given matrix via the singular value
decomposition (SVD); see \cite[\S8.6]{GolLoa13}.  We will assume in
this section that $V\in\RR[t]^{N\times N}$ is the differential
Sylvester matrix from Subsection \ref{ssec:dsyl}, of differential
polynomials $f,g\in\RR[t][\D;']$ of degrees (in $\D$) of $m$ and $n$
respectively, with $N=m+n$.  From this we construct the inflated
differential Sylvester matrix $\Vhat\in\RR^{N(\mu+1)\times
  N(\mu+d+1)}$ as in Subsection \ref{ssec:idsyl}.  Using the SVD we
can find the matrix $\Delta\Vhat$ of minimal 2-norm such that $\Vhat +
\Delta\Vhat$ has a prescribed rank.  First, we compute the SVD of $\Vhat$
as
\[
\Vhat = P\Sigma Q,
\]
where 
\[
P\in \RR^{N(\mu+1) \times N(\mu+1)}, ~~\mbox{and}~~ Q\in \RR^{N(\mu
  +d+1) \times N (\mu+d+1)}
\]
are orthogonal and
\[
\Sigma = \diag(\sigma_1,\ldots, \sigma_{N(\mu+1)}) \in
\RR^{ N(\mu+1)\times N(\mu+d+1)},
\]
satisfies $\sigma_1 \geq \sigma_2 \ldots \geq \sigma \geq
\sigma_{N(\mu+1)}$.  Note that $\Sigma$ is not square (it has more
columns than rows), and we simply pad it with zeros to obtain the
desired shape.

Now, by Lemma \ref{lem:RRGCRD}, we want to find a nearby matrix whose
left nullspace has reduced dimension by multiples of $(\mu+d+1)$, that
is
\[
P\, \Sigmabar\, Q = \Vhat + \Delta \Vhat,
\]
where
\begin{align*}
\Sigmabar & =
\diag(\sigma_1,\sigma_2,\ldots,\sigma_{(N-\varrho)(\mu+d+1)},0,\ldots,0)\\ 
& \in
\RR^{ N(\mu+1)\times N(\mu+d+1)},
\end{align*}
where $\varrho=\frac{\dim\lnullspace(\Vhat+\Delta\Vhat)}{\mu+d+1}$.
Then $\Vhat$ will be by the singular matrix, $\Vhat+\Delta\Vhat$ of prescribed
rank.  Of course, $\Vhat+\Delta\Vhat$ is probably an unstructured
matrix, and in particular, not an inflated differential Sylvester matrix.

Next we show that a matrix of the desired rank deficiency and (inflated
differential) structure exists within a
relatively small radius of $\Vhat$. Suppose there is an
$\ftil,\gtil\in\RR[t][\D;']$, with $\norm{\ftil-f}\leq\epsilon$ and
$\norm{\gtil-g}\leq \epsilon$, such that
$\degD\gcrd(\ftil,\gtil)=\varrho\geq 1$.  Let $\Delta f=f-\ftil$ and $\Delta
g=g-\gtil$, so $\norm{\Delta f}, \norm{\Delta g}<\epsilon$.
Moreover, the differential resultant matrix $W\in\RR[t][\D;']^{N\times
  N}$ formed from $\Delta f$, and $\Delta g$ has
\[
\norm{W}_F^2<(d^{2n}+d^{2m})\cdot\epsilon^2
\]
by Lemma \ref{lem:vandnorm}. Thus, the inflated differential resultant
matrix $\What\in\RR^{N(\mu+1)\times N(\mu+d+1)}$ has
\[
\norm{\What}_2^2<\mu\cdot(d^{2n}+d^{2m})\cdot\epsilon^2
\]
by Lemma \ref{lem:infvandnorm}.  Moreover,
$\dim\lnullspace(\Vhat+\What)=\varrho(\mu+d+1)$.  Thus, for
sufficiently small $\epsilon$ there exists a perturbation $\What$ such
that $\norm{\What}$ is small (at least assuming small $n,m$) and
$\Vhat+\What$ has appropriate rank and structure.

Due to the unstructured nature of $\Vhat + \Delta \Vhat$, one must
take care in working with a reasonable approximation for $f$ and $g$.
However, there is considerable redundancy of the coefficients of $f$
and $g$ in their inflated differential Sylvester matrix, if only
because each entry of $f$ and $g$ shows up multiple times under the
map $\Gamma$; see Section \ref{ssec:idsyl}.  There is, in fact, even more redundancy
because of the different derivatives in rows of the differential
Sylvester matrix, but we will not capitalize on this. 

\subsection{Computing {\large$\boldmath V+\Delta V$} and
  reconstructing {\large $\boldmath f+\Delta f$} and
  {\large $\boldmath g+\Delta g$ }}
\label{ssec:recovery}


To form $V + \Delta V$ we take  a weighted
average of the $t$-shifted blocks of $\Vhat+\Delta \Vhat $ that
correspond to $f$ and $g$. This involves identifying the ``blocks''
of $\Vhat+\Delta \Vhat$ that correspond to $f$ and $g$ and re-constructing them
entrywise ensuring the entries in degree $t$ of $f$ and $g$ do not increase.
The $f$ block consists of rows $1$ through $\mu+1$ and the $g$ block consists of
rows $\degD g(\mu+1)+1$ through $\degD g  (\mu+1) + \mu +2$. The columns in both
cases are the $\mu+d+1$ columns for each block entry. 

The reason that this reconstruction is often satisfactory
is we have that if $\varrho>0$ then $\Vhat$ has rank
$(N-\varrho)(\mu+d+1)$, and if $\sigma_{(N-\varrho)(\mu+d+1)+1} < \epsilon$, then 
\begin{align*}
\norm{\Sigma - \Sigmabar }_F^2 & =  \sum_{i= (N-\varrho)(\mu+d+1)+1}^{(\mu+1)N} |\sigma_i|^2 \\ 
			        &  \leq \epsilon^2 \left [(\mu+1)N - (N-\varrho)(\mu+d+1)\right]  \\
			        &  \leq \epsilon^2 \varrho(\mu+d+1).
\end{align*}
We have that $\norm{\What}^2_F  =   \norm{\Sigma - \Sigmabar}^2_F$ \citep[Corollory 2.4.3]{GolLoa13}
because the  singular values of $\What = P(\Sigma - \Sigmabar)Q$
are a permutation of the entries along the main diagonal of $\Sigma-\Sigmabar$.  

In our construction we require that $\deg_t
\ftil_i \leq \deg_t f_i$ for $0\leq i \leq \degD f$ and a similar
condition on $g$ in order to preserve the structure of $V+\Delta
V$. Furthermore, if the perturbation from adjusting the singular
values is small, then the non-zero entries are ``small'' and can
usually be ignored without losing too much information.

We should now have a matrix that is numerically singular, $V+\Delta V$
and perturbations $\Delta f$ and $\Delta g$ such that
$\mathtt{NumericGCRD}(f+\Delta f, g+\Delta g)$ is non trivial
and satisfies the conditions $\deg_t \Delta f_i \leq \deg_t f_i$
for $0\leq i\leq m$ and $\deg_t \Delta g_j \leq \deg_t g_j$ for 
$0\leq j \leq n$.

\subsection{Computing the Approximate GCRD}
\label{sec:ApproxGcrd}
Let $f,g \in \RR[t][\D;']$ have degrees $m$ and $n$ respectively. Let
$G=\gcrd(f,g)$ and $\degD(G) = D$. Then one may obtain an $\RR(t)$
multiple of $G$ by solving
\[
wV = \begin{pmatrix} 
  *_0 & *_1 & \cdots & *_D & 0 & \cdots & 0 \end{pmatrix},
\]
where we do not care about the entry $*$.  Solving this system will
give us a multiple of $G$, which we may assume is in $\RR[t]$ by
clearing fractions from the denominator.

\subsubsection*{Computing an Approximate Primitive GCRD}
When computing the GCRD numerically we obtain a result that is an
$\RR[t]$-multiple of a primitive GCRD upon clearing fractions.  In
some applications it is desirable to remove this content.  However the
coefficients are not exactly known so taking an exact GCD of the
coefficients will yield unsatisfactory answers. 

Consider the case of our \texttt{NumericGCRD} algorithm, Algorithm~\ref{alg:NumericGCRD}.
Our solution will have approximate content even if we use rational arithmetic
because $f+\Delta f$ and $g+\Delta g$ have an approximate GCRD but may
not have a GCRD algebraically due to round off errors in their
recovery.  If a primitive solution is desired, then techniques to
remove the content are required.

\subsubsection*{Leading Coefficient Known in Advance}
Since leading coefficients of GCRDs are propagated through
multiplication, we often know the leading coefficient of a GCRD in
advance. In general, the leading coefficient of the 
GCRD should be approximately one (i.e., a constant)
save a few special cases where the leading coefficients of
$f$ and $g$ have a non-trivial approximate GCD.

As an observation, given $f,g$ in $\RR[t][\D;']$ where $\degD f = m$ and $\degD g = n$
and $\gcd(f_m,g_n) =1$ then for a suitable approximate GCD algorithm 
(see \cite{Corless04} or \cite{ZenDay04})
we have that a primitive numeric GCRD of $f$ and $g$ satisfies
$\lcoeff(G) = 1$. The reason that we need the notion of approximate GCD
is that it is possible that we may have different algorithms returning different
answers. Consider $t^2$ and $t^2+2^{-s}$ where $s$ is large. Algebraically
both polynomials are co-prime but if $s$ is sufficiently large then some 
GCD algorithms will return a non trivial GCD.
To justify this, let $\degD G = D$. $G$ is a GCRD so we have that $G$ divides both
$f$ and $g$ on the right.  If
$G$ is primitive, it follows that $G_D | f_m$ and $ G_D | g_n$ so
$G_D | \gcd(f_m,g_n) = 1$. This occurs if and only if $G_D = 1$.

If we are given a candidate GCRD $\widetilde G$ that is not primitive,
where $\degD \widetilde G = D$, and we know the primitive GCRD has
leading coefficient 1, then $\lcoeff(\widetilde G) \approx \content
\widetilde G$.  It follows that $\widetilde G_D | \widetilde G_i$ for
$0\leq i <D$. This means that the remainders are numerically trivial,
so we can assume they are zero.  Using a method of approximate
division we can recover a primitive approximate GCRD.

 If the primitive GCRD is known to be $1$, it is not sufficient to
 solve the system
 \[ 
  wV = \begin{pmatrix} *_0 & *_1 & \hdots & *_{D-1} & 1 & 0 \hdots
   0 \end{pmatrix} ^T  
 \]
 because we are performing numerical linear algebra and we will often
 obtain a solution over $\RR(t)[\D;']$.
 One particular method of approximate division is by interpolation.
 This particular method of approximate division yields answers
 that one would expect with prior knowledge of a GCRD
 and a uniform distribution of noise. As expected with an interpolation
 based method, it will break down if we are unable to accurately compute
 the degrees of terms due to artifacts and round off errors. 
 The method of \cite{Bini86} for approximate division via a Fast
 Fourier Transform (FFT) is used for content removal in this case,
 and proves both fast and numerically robust in practice.

  \begin{algorithm}
   \label{alg:cont-remove}
   \caption{ Content Removal via Approximate Division}
   \begin{algorithmic}[1]
     \Require $\widetilde G \in \RR[t][\D;']$ of degree $D$ in $\D$ with
     appropriate degree of entries (leading coefficient has minimal
     degree).  
     \Ensure $G \approx \widetilde G / \content (\widetilde G)$.
   
     \State Vectorize $G$ as $(G_0,G_1,\ldots,G_D)$.
     \For { $i$ from
	    $0$ to $D-1$ } Remove left over artifacts from our linear algebra below
	    a given threshold. 
      \State Compute $\mathtt{FFT}(\widetilde G_i)$ and
	      $\mathtt{FFT}(\widetilde G_D)$ using a $\deg_t(\widetilde G_i) +1$ root of
	      unity. 
     \State Compute the element wise quotient $\widetilde G_i/
	    \widetilde G_D$ for each entry.
    \State Compute $G_i =
	    \mathtt{InverseFFT}(\widetilde G_i/ \widetilde G_D)$
    \State Set the
	    last $\degD$ terms to $0$ and remove other artifacts from the FFT
	    below a given threshold.
     \EndFor
     \State set $G_D = 1$ 
     \State Devectorize $(G_0,G_1,\ldots,G_D)$
     into $G$. 
     \State  \Return $G$
   \end{algorithmic}
 \end{algorithm}

\begin{example}[Numeric GCRD]

Consider $$f = - 0.45\,{{\it \D}}^{2}- 0.56\,t{\it \D}- 0.11\,{t}^{2}- 0.45$$ and 
$$g ={{\it \D}}^{3}+ \left( t+ 0.66 \right) {{\it \D}}^{2}+ \left(  2.0+
 0.952\,t \right) {\it \D}+ 0.66+ 0.292\,{t}^{2}.$$ 

A numeric GCRD of $f$ and $g$ is given by 

\begin{align*}
G &=     \left(  0.02781\,{t}^{2}+ 0.30990\,t- 0.01460\,{t}^{3}- 0.11380
	 \right) {\it \D} \\
	 +&  0.02781\,{t}^{3}+ 0.30993\,{t}^{2}- 0.01461\,{t}^
	 {4}- 0.11378\,t- 0.00002. 
\end{align*}
Given the low degrees and leading coefficients, a primitive GCRD is probably a unit in $\RR$.
Removing content with an FFT gives us a primitive numeric GCRD of $ {\it \D}+1.00000\,t$.
\end{example}

\subsubsection*{Leading Coefficient of {\large $\boldmath{G}$} Unknown}

It is not possible to determine the leading coefficient of a primitive GCRD
in advance when the leading coefficients of $f$
and $g$ have a non-trivial GCD, or if they share a nearby common solution.
In order work around this one would need to approximate 
$\content(G) = \gcd(G_0,G_1,\ldots, G_D)$ numerically, then perform an approximate
polynomial division.  We used the method of \cite{Corless04} to compute pair-wise GCDs
and obtained somewhat mixed results. In some instances the content had a degree
that was too small based on our construction or our division algorithm
did not provide an answer consistent with the GCRD we constructed. 
We might hope to overcome some of these problems using a more specific method for the
GCD of multiple polynomials, as developed in \cite{KYZ05}.
\pagebreak

\subsection{Algorithms}
This section provides a high-level description of the primary algorithms we are using,
in pseudocode. 

In practice, we will demonstrate that our algorithms work well on low
degree differential polynomials as input.  However, in general our algorithms
are not provably \emph{guaranteed} to give  close polynomials with a
non-trivial GCRD.  In practice, if $\norm{f} = \norm{g}=1$ and
$\norm{\Delta f} = \norm{\Delta g} < 0.1$
then for low degree differential polynomials we are usually able to get 
an answer that is acceptable based on our prior knowledge of $\gcrd(f,g)$. 

\begin{algorithm}[!h]
  \caption{\bf : NumericGCRD}
  \label{alg:NumericGCRD}

 \begin{algorithmic}[1]
   \smallskip
   \Require 
      \item[$\bullet$] $f,g \in \RR[t][\D;']$ non-zero with $\norm{f}=\norm{g}=1$;
     \item[$\bullet$]
       A search radius $\epsilon>0$.
       \smallskip
   \Ensure $G \approx \gcrd(f,g) \in \RR[t][\D;']$ with $\degD G\geq
   1$, or an indication that $f$ and $g$ are co-prime within search
   radius $\epsilon$.
   \medskip
   \State $m\gets \degD f$, $n\gets \degD g$, $d\gets \max\{\deg_t
   f,\deg_t g\}$ and $\mu \gets 2(m+n)d$.
   \State Form the differential Sylvester matrix\newline
             $V(f,g) \in \RR[t]^{(m+n)\times(m+n)}$.
   \State Form  the inflated differential Sylvester matrix\newline
            $\Vhat = \Vhat(f,g) \in \RR^{ (m+n)(\mu+1) \times (m+n)(\mu+d+1)}$ of $V$.
   \State Compute the numerical rank $r$ of $V$ using Algorithm~\ref{alg:RankAlg}
	  on $\Vhat$ with search radius $\epsilon$.
   \State If $r>0$ then set $D = m+n-r$. Otherwise indicate that $f$ and $g$ 
	  are co-prime with respect to $\varepsilon$ and return.
   \State Solve for $w$ from
	  \[
	  wV = \begin{pmatrix} *_0 & *_1 & \hdots & *_D & 0 & \hdots & 0
	  \end{pmatrix} ^T.
	  \]
   \State Set $G = wV$ 
   \State Optionally remove the content from $G$ numerically.
 \end{algorithmic}
\end{algorithm}

\begin{algorithm}[!h]
  \caption{\bf : Nearest With GCRD}
  \label{alg:SVD-GCRD}

 \begin{algorithmic}[1]
   \smallskip
   \Require 
   \item[$\bullet$] $f,g \in \RR[t][\D;']$ with $\norm{f}=\norm{g}=1$;
   \item[$\bullet$] A search radius $\epsilon>0$.
   \Ensure $f+\Delta f$, $g+\Delta g$  where $\deg_t \Delta f\leq \deg_t f$, $\degD \Delta f \leq  \degD f$, $\deg_t \Delta g \leq \deg_t g$, $\degD 
   \Delta g \leq \degD g$ and \newline
   $G \approx \gcrd(f + \Delta f,g + \Delta g) \in
	    \RR[t][\D;']$ with $\degD G\geq 1$, \linebreak
  or an indication that $f$ and $g$ are co-prime within search
   radius $\epsilon$.

   \medskip
   \State $m\gets \degD f$, $n\gets \degD g$, $d\gets \max\{\deg_t
   f,\deg_t g\}$ and $\mu \gets 2(m+n)d$.
   \State Form the differential Sylvester matrix\newline
	  $V(f,g) \in \RR[t]^{(m+n)\times(m+n)}$.
   \State Form  the inflated differential Sylvester matrix \newline
         $\Vhat = \Vhat(f,g) \in \RR^{ (m+n)(\mu+1) \times (m+n)(\mu+d+1)}$ of $V$.
   \State Compute the SVD of $\Vhat$,  $\Vhat = P\Sigma Q$ with $P,\Sigma$ and $Q$
	  as discussed in \S\ref{sec:lsquare}.
   \State Compute  the numerical rank $r$ of $V$ using Algorithm~\ref{alg:RankAlg}
	  on $\Vhat$ with search radius $\epsilon$. 
  \State  If $r>0$ set the last
	  $r(\mu+d+1)$ singular values to $0$ and compute $\Sigmabar$
          as discussed in  \S\ref{sec:lsquare}.
	  Otherwise indicate that $f$ and $g$ are co-prime with
          respect to $\epsilon$ \hbox to 0pt {and return.}
  \State Compute  $\Vhat + \Delta \Vhat = P \Sigmabar Q$.
  \State Recover $f+\Delta f$ and $g+\Delta g$ from $\Vhat + \Delta \Vhat$
	 as discussed in \S\ref{ssec:recovery}.
  \State Compute $G = \mathtt{NumericGCRD}(f+\Delta f,g + \Delta g)$ using Algorithm \ref{alg:NumericGCRD}
	 with $\epsilon$ used to validate the degree of our approximate GCRD. 
  \State Optionally remove the content from $G$ numerically.

 \end{algorithmic}

\end{algorithm}
As an observation, for  $f,g \in \RR[t][\D;']$ of 
degrees $m$ and in $n$ in $\D$ respectively, we can use
some heuristics to detect failures of our algorithms. 
Let $G = \gcrd(f,g)$ and $\degD G = D$.
It is clear  that if $D > \min \{n,m\}$
 then $G$ cannot be a GCRD. Simmilarily 
 if $\deg_t(G_i) < \deg_t(G_D)$ for some $
  0\leq i < D$ then $G$ cannot be a GCRD of $f$ and
  $g$ over $\RR[t]$ if 
  $\gcd( \lcoeff(f), \lcoeff(g)) = 1.$

\begin{algorithm}[!h]
  \caption{\bf : Deflated Rank}
  \label{alg:RankAlg}

 \begin{algorithmic}[1]
   \smallskip
   \Require 
   \item[$\bullet$] An inflated differential Sylvester
     matrix\newline
         $\Vhat \in \RR^{(m+n)(\mu+1) \times
       (m+n)(\mu+d+1)}$ with $m,n,d$ and $\mu$ defined as in
	     Algorithm~\ref{alg:NumericGCRD} or
             Algorithm~\ref{alg:SVD-GCRD}.
   \item[$\bullet$] 
	     A search radius $\epsilon >0$ 

  \smallskip
   \Ensure The the numeric rank of the (non-inflated) differential
   Sylvester matrix $V$ from Algorithm~\ref{alg:NumericGCRD}  or \ref{alg:SVD-GCRD}. 

   \medskip
   \State Find the maximum $k$ such that $\sigma_k > \epsilon
	  \frac{\sqrt{ (m+n)(2\mu+d+2) }}{\mu+d+1}$ and
	  $\sigma_{k+1}<\epsilon$.
   \State if $\sigma_k > \epsilon$ for
	  all $k$ then $\Vhat$ has full rank.
   \State If there is no significant change between $\sigma_k$ and $\sigma_{k+1}$ for all $k$
	  as determined by step 2 then return failure. 
   \State Set $ r =\left \lceil \frac{k}{\mu+d+1} \right \rceil$, the scaled rank.
 \end{algorithmic}
\end{algorithm}

\begin{example}[Nearest With GCRD]\label{ex:svd-gcrd}
  Consider
  \[
  f=\left(  1.0+ 0.0043\,t \right) {{\it \D}}^{2}+ \left(  3.0\,t- 0.0003
    \right) {\it \D}+ 2.0\,{t}^{2}+ 1.0
  \]
  and
  \[
  g = {t}^{2}{{\it \D}}^{2}+ \left( - 0.0004\,t+{t}^{3}+ 0.0001 \right) {
    \it \D}+{t}^{2}
  \]
  with a given search radius $\epsilon = .5 \cdot 10^{-2}.$ Applying 
  Algorithm ~\ref{alg:RankAlg} on the singular values of $\Vhat$ we obtain 
  a rank of $78$ where the expected rank is $84$, which is within reason.
  From this we conclude that the degree of our GCRD is $1$ and we compute
 \begin{align*}
   \ftil = & 0.99999+ 2.00000\,{t}^{2} \\
   + & \left( - 0.00029+  3.00000\,t \right) {\it \D} \\
     + &  \left( 0.99999+ 0.00429\,t \right)
   {{\it \D}}^{2}
 \end{align*}
 and
\begin{align*}
  \gtil = & - 0.00001- 0.00002\,t+ 1.00001\,{t}^{2} \\ + & \left(
    0.00011- 0.00039\,t+ 0.00005\,{t}^{2}+ 0.99999\,{t}^{3} \right)
  {\it \D} \\ +& \left( - 0.00001- 0.00002\,t+ 0.99999\,{t}^{2}
  \right) {{\it \D}}^{2}.
\end{align*}

Furthermore, we obtain $\mathtt{NumericGCRD}(\hat f, \hat g) \approx
 -0.06695+1.06508 t+\D$, after removing content.  We have that the size of the
perturbations are $\| f-\ftil \| = 1.22023 \times 10^{-9}$ and $\| g- \gtil \| =0.00007$.
In this example the largest singular value we removed was known to three decimal places so we can expect
the accuracy of our answer to reflect this. 
\end{example}
As a remark, we constructed $\ftil$ and $\gtil$ from their first occurrence in $\Vhat$ for illustration purposes.
In general weighted approaches work well and are demonstrated in Section \ref{sec:exp-data}.

\subsubsection*{On the Rank Algorithm}
Algorithm \ref{alg:RankAlg} follows similarly to the method for
determination of numerical rank in the SVD GCD of \cite{CGTW95}.  In
our circumstances the singular values are in (column) ``blocks'' of size $\mu+d+1$.
We choose to declare the scaled rank as
$ \lceil  \frac{k}{\mu+d+1}\rceil$ because it will tend to ignore spurious singular
values provided there are  fewer spurious singular values than the block size. Since we know
that non-spurious singular values come in the block size, this allows us to accurately compute the rank of 
$V$ even if we don't correctly compute the rank of $\Vhat$ in the case of low degree
differential polynomials.

Another important reason for this decision is if we set a singular value to $0$
that is reasonably far away from the next singular value then the resulting matrix $\Vhat + \Delta \Vhat$
becomes highly unstructured and our algorithm fails to produce meaningful results.



It is possible that the singular values are not clearly
separated. This happens frequently with dense large degree
differential polynomials. In this scenario it is possible to under
estimate the rank of $V$.  If $f$ and $g$ are known to have a GCRD,
then it may be possible to eliminate bad GCRD candidates using heuristics.


 \section{Complexity \& Stability Analysis}

In this section we investigate the computational complexity in terms
of the input size and some of the numerical stability of the
algorithms. We assume as usual that for $f,g \in \RR[t][\D;']$
that $\degD f = m$, $\degD g = n$, $V= V(f,g)$\linebreak $\in \RR[t]^{
  (m+n) \times (m+n)}$ and $\Vhat \in \RR ^{(m+n)( \mu+1) \times
  (m+n)(\mu +d+1)}$ with $\deg V = d$. We assume that all operations
over $\RR$ can be performed in a constant amount of time, i.e., are
floating point operations or \emph{flops}.  Our
algorithms will be analyzed in a bottom-up approach, reflecting their
dependencies.

\subsubsection*{Analysis of Algorithm  \ref{alg:RankAlg}: {\large
    \texttt{Deflated Rank}}}
    
The cost of the algorithm is dependent on computing the singular
values of $\Vhat$. The cost of the computing the singular values of
$\Vhat$ is $O( (m+n)^6d^3)$ operations over $\RR$ or flops, using the
standard method of matrix multiplication.

\subsubsection*{Analysis of Algorithm \ref{alg:NumericGCRD}: 
  {\large\texttt{NumericGCRD}}}

The cost of the algorithm is dominated by that of computing the singular
values of $\Vhat$ to determine the rank of $V$.
The cost of performing the linear algebra on $V$ is 
$O(m+n)^3 d^2)$ operations over $\RR$.   The cost of the SVD
on $\Vhat$ is $O( (m+n)^6d^3)$ operations and hence dominates.

In general the numerical stability of the algorithms depends on the
methods used to solve the problem and the conditioning of $V$. It is
difficult to say much else without making further assumptions.  We
hope to explore this further in subsequent work.

\subsubsection*{Analysis of Algorithm  \ref{alg:SVD-GCRD}: {\large
    \texttt{Nearest With GCRD}}}

The cost of the algorithm is dominated by that of computing the singular
values of $\Vhat$ to determine the rank of $V$. This
requires $O(((m+n)^6d^3)$ operations over
$\RR$ or flops, with the usual method of matrix multiplication. This algorithm
will call Algorithm ~\ref{alg:NumericGCRD}, but again the cost
of computing the singular values of $\Vhat$ is the dominating cost.

The accuracy of our answer is subject to the singular values we set to
zero or remove.  In particular, the larger the singular values we
remove, the less accurate our answer typically becomes. In our experiments this
became noticeable with singular values around $10^{-3}$.  Under the
assumption that noise is distributed uniformly over $f$ and $g$ then
the singular values of $\Vhat$ will generally remain small, which means
Algorithm~\ref{alg:SVD-GCRD} typically produce a reasonable numeric
GCRD.  Again, we hope to bolster these heuristic observations with
more careful analysis in subsequent work.


 \section{Experimental Evaluation}
\label{sec:exp-data}

In order to verify the robustness of the algorithm and whether it is
able to compute $\ftil$ and $\gtil$ reliably, we performed 100
random trials with different search neighborhoods and adding random
amounts of noise.  In our experiments we are adding a noise factor of
size $\delta$ and working with a search radius of size $\rho$.
The goal of our experiments is to demonstrate that given a pair
of relatively prime polynomials $\fhat$, $\ghat \in \RR[t][\D; ']$ where 
$\norm{f-\fhat} = \norm{g-\ghat} = \delta$ for some $f$ and 
$g$ such that $\gcrd(f,g)$ is non trivial,
we can recover a pair of polynomials such that Algorithm~\ref{alg:SVD-GCRD}
returns a non trivial answer in a given search radius
of size $\rho$.  More precisely, the perturbations
$\norm{\fhat - \ftil}$ and  $\norm{\ghat - \gtil}$ are
minimal approximations to the 2-norm we are using in our
least squares setting. 
We can think of $\varepsilon$ from Problem~\ref{prb:main-prob}
as being quantified by $\max\{ \norm{\fhat - \ftil} ,\norm{\ghat - \gtil} \}$
in these tests.

We perform tests on two sets of examples where the input size was
bounded.  We found the algorithm worked quite well in practice on
examples taken uniformly at random, however the lack normalization or
structure made it difficult to obtain comparable data. In the data
tables we provide statistics for two different reconstructions.  The
first approach is reconstructing $\ftil$ and $\gtil$ from the first
row they appear in $\Vhat+\Delta \Vhat$. The other reconstruction
approach is using the weighted approach over the entire block to
recover $\ftil_w$ and $\gtil_w$ respectively. On average the weighted
approach tends to smooth values over at the cost of structure where as
taking the first row can preserve the underlying structure of $f$ and
$g$, especially if they are sparse.

We note that as the noise decreases so does the size of the
perturbation. We will continue to get smaller perturbations until the
roundoff error from writing $\Vhat \approx P\Sigma Q$ dominates the
perturbation sizes.  If the noise is sufficiently small then we are
executing \texttt{NumericGCRD} (Algorithm~\ref{alg:NumericGCRD}) as
the perturbations will become indistinguishable from roundoff errors.
Despite the noise in some tests 
appearing large, it is distributed uniformly into the coefficients resulting
in each coefficient being perturbed by a fraction of the total amount of
noise added. 
Although the worst case perturbations in $\ftil$ and $\gtil$ were 
relatively large, such perturbations were uncommon in our experiments
and we still managed to obtain
valid candidates.  in the context of Problem~\ref{prb:main-prob}.
We justify adding noise uniformly since in practice
we would expect this work to be applied on data which suffers from
round off errors which tend to be uniformly distributed across the
data. 

\subsection{Bounded Coefficient Tests}
In this section we perform our tests on $f$ and $g$ whose coefficients
are bounded and somewhat structured.  The following steps detail the
construction of our examples.
\begin{enumerate}
\item Generate $h_1,h_2,h_3 \in \RR[t][\D;']$ where $\deg_t h_i \leq
  2$ and $ 1 \leq \degD h_i \leq 2$.
\item Set $h_{ij} = h_{ij}/\|h_{ij} \|$ for $i=1,2,3$ and $1\leq j
  \leq \deg h_i$.  More precisely, $\|h_i\| = \deg_t h_i +1$.
\item Compute $f= h_1 \cdot h_3 + \delta_f$ and $g=h_2\cdot h_3 +
  \delta_g$, where the noise is distributed uniformly.
\item Run Algorithm~\ref{alg:SVD-GCRD}, the \texttt{Nearest With GCRD} algorithm on $f$
  and~$g$.
\item If \texttt{Nearest With GCRD}$(f,g) =1$ then ignore the result in the statistics.
\end{enumerate}

The construction of these examples seems peculiar and
counterintuitive, but it provides a suitable set of tests.  The
justification for the set of tests is to get an idea how the algorithm
performs on random data that is not normalized, but is bounded with a
random structure.

The table in Figure~\ref{fig:perturbations-bounded} details the
relevant statistics obtained from running
Algorithm~\ref{alg:SVD-GCRD}.  The table in Figure~\ref{fig:failures-bounded}
provides the number of trivial GCRDs that occurred for a given
$\rho$ and $\delta$.
 
\begin{figure}[h]

  \caption{Perturbation statistics for bounded coefficients}
  \center  \label{fig:perturbations-bounded}
  \begin{tabular}{|c|c|c|c|c|} \hline
    $(\rho,\delta, \text{ {\scriptsize Reconstructed}})$  & Max & Average & {\scriptsize Standard Deviation}\\ \hline
    $(.5,.5,\ftil)$ 		& 0.269443 & 0.022653 & 0.040056  \\
    $(.5,.5,\gtil)$ 		& 0.233732 & 0.025051 & 0.037960 \\
    $(.5, .5, \ftil_w)$ 	& 0.269443 & 0.023507 & 0.038163  \\
    $(.5, .5, \gtil_w)$		& 0.233732 & 0.023625 & 0.032776  \\  \hline
    $(.5, .1,\ftil)$ 		& 0.047198 & 0.010184 & 0.011120 \\
    $(.5, .1,\gtil)$ 		& 0.061742 & 0.011490 & 0.013434  \\
    $(.5, .1, \ftil_w)$	 	& 0.045457 & 0.009629 & 0.010101  \\
    $(.5, .1, \gtil_w)$		& 0.049726 & 0.010082 & 0.010634  \\  \hline
    $(.5, .01,\ftil)$ 		& 0.233401  & 0.006653 & 0.030810 \\
    $(.5, .01,\gtil)$ 		& 0.166854 & 0.005394 & 0.021614  \\
    $(.5, .01, \ftil_w)$ 	& 0.233401 & 0.005647 & 0.027110  \\
    $(.5, .01, \gtil_w)$	& 0.166854 & 0.005647 & 0.018869  \\  \hline
     $(.5, .001,\ftil)$ 		& 0.283611  & 0.006625 & 0.037541 \\ 
     $(.5, .001,\gtil)$ 		& 0.182687 & 0.004582 & 0.025346  \\
     $(.5, .001, \ftil_w)$	 	& 0.283611 & 0.006277 & 0.037003  \\
     $(.5, .001, \gtil_w)$		& 0.182687 & 0.004317 & 0.024933  \\  \hline
    \end{tabular}
\end{figure}

 \begin{figure}[h]

 \caption{Trivial Numeric GCRDs for Bounded Coefficients}
 \center   \label{fig:failures-bounded}
    \begin{tabular}{|c|c|c|c|} \hline 
    $\rho$  & $\delta$&  Trivial GCRD \\ \hline
    .5 		& .5	    	& 9 \\
    .5 		& .1	    	& 3  \\
    .5 		& .01	    	& 3  \\
    .5 		& .001	    	& 0  \\\hline
    .05	 	& .5	    	& 95  \\
    .05		& .1	    	& 42  \\
    .05		& .01	    	& 0  \\
    .05		& .001	    	& 2  \\\hline
    \end{tabular}  
 \end{figure}
\subsection{Normalized Tests}
In this section we perform our tests on $f$ and~$g$ that were
normalized and added noise.  The following steps detail the
construction of our examples.
\begin{enumerate}
\item Generate $h_1,h_2,h_3 \in \RR[t][\D;']$ where $\deg_t h_i \leq
  2$ and $ 1 \leq \degD h_i \leq 2$.
\item Compute $f= \frac{h_1 \cdot h_3}{\|h_1 \cdot h_3 \|} + \delta_f$
  and $g=\frac{h_2 \cdot h_3}{\|h_2 \cdot h_3 \|} + \delta_g$, where
  the noise is distributed uniformly.
 \item Run Algorithm~\ref{alg:SVD-GCRD}, the \texttt{Nearest With GCRD} algorithm on $f$ and~$g$.
 \item If $\texttt{Nearest With GCRD}(f,g) =1$ then ignore the result in the statistics.
\end{enumerate}

The table in Figure~\ref{fig:perturbations-normalized} details the
relevant statistics obtained from running
Algorithm~\ref{alg:SVD-GCRD}.  The table in Figure~\ref{fig:failures-normalized} provides
the number of trivial GCRDs that
occurred for a given $\rho$ and $\delta$.
 
\begin{figure}[!h]
  \caption{Perturbation statistics for normalized $f$ and $g$}
  \center   \label{fig:perturbations-normalized}
  \begin{tabular}{|c|c|c|c|c|} \hline
    $(\rho,\delta, \text{ {\scriptsize Reconstructed}})$  & Max & Average & {\scriptsize Standard Deviation} \\ \hline
    $(.5,.5,\ftil)$ 		& 0.108768 & 0.015880 & 0.022180 \\ 
    $(.5,.5,\gtil)$ 		& 0.158056 & 0.016359 & 0.024371 \\
    $(.5, .5, \ftil_w)$ 	& 0.052060 & 0.010584 & 0.010404  \\
    $(.5, .5, \gtil_w)$		& 0.097296 & 0.011577 & 0.014039  \\  \hline
    $(.5, .1,\ftil)$ 		& 0.058105  & 0.017043 & 0.014834 \\
    $(.5, .1,\gtil)$ 		& 0.063637 & 0.025312 & 0.017561  \\
    $(.5, .1, \ftil_w)$	 	& 0.057251 & 0.012524 & 0.013178  \\
    $(.5, .1, \gtil_w)$		& 0.046134 & 0.019487 & 0.012614  \\  \hline
    $(.5, .01,\ftil)$ 		& 0.053014  & 0.005987 & 0.009114 \\
    $(.5, .01,\gtil)$ 		& 0.057982 & 0.006851 & 0.009757  \\
    $(.5, .01, \ftil_w)$	& 0.031045 & 0.004263 & 0.006710  \\
    $(.5, .01, \gtil_w)$	& 0.038133 & 0.005996 & 0.007523  \\  \hline
    $(.5, .001,\ftil)$	 	& 0.066271  & 0.003778 & 0.009018 \\
    $(.5, .001,\gtil)$ 		& 0.093175 & 0.004416 & 0.011236  \\
    $(.5, .001, \ftil_w)$	& 0.067082 & 0.003556 & 0.009915  \\
    $(.5, .001, \gtil_w)$	& 0.038098 & 0.003213 & 0.006429  \\  \hline
  \end{tabular}
\end{figure}

\begin{figure}[!h]
  \caption{Trivial Numeric GCRDs for Normalized $f$ and
    $g$}
  \label{fig:failures-normalized}
 
  \center 
  \begin{tabular}{|c|c|c|} \hline 
    $\rho$  & $\delta$& Trivial GCRD \\ \hline
    .5 		& .5	   	& 6 \\
    .5 		& .1	    	& 0\\
    .5 		& .01	    	& 2\\
    .5 		& .001	    	& 1\\\hline
    .05	 	& .5	    	& 86\\
    .05		& .1	    	& 22\\
    .05		& .01	    	& 2 \\
    .05		& .001	    	& 1\\\hline
    \end{tabular}  
 \end{figure}


 \pagebreak	
 \section{Conclusions and Future Work}

We have developed a framework for approximate differential polynomials
and demonstrated an algorithm for computing the greatest common right
divisor of two approximate differential polynomials.  This corresponds
to finding a representation of the common solutions of two approximate
linear differential operators.

Our algorithm makes use of the (unstructured) SVD approach introduced
in \cite{CGTW95}, and works well when differential polynomials with a
non-trivial GCRD are nearby.  Unfortunately, it also suffers some of
the same drawbacks, especially when the nearest polynomials with a
non-trivial GCRD are relatively far away.  Even more so, the lack of a
geometric root space (as for conventional polynomials) makes the
analysis substantially more difficult.  One possible remedy we are
exploring is to this is to use a structured matrix approach, as explored
in \cite{Beckerman97}, \cite{BGM05}, \cite{KalYan06}.  In particular,
Riemannian SVD-like methods would seem relatively easy to adapt in
this case (though again, analysis will not be easy).

Another problem is the factorial-like scaling introduced by multiple
differentiations when constructing the (inflated) differential
Sylvester matrix.  This leads to numerical instability when the degree
in $\D$ gets even modestly large.  Some row scaling may alleviate this
somewhat, but an alternative differential resultant formulation would
seem a better approach overall, and is a path we are investigating.

The algorithms described in this paper have been implemented in Maple
and are primarily based on the LinearAlgebra package.  This allows
flexibility in determining the method used to solve the linear systems
and dealing with other problems such as content removal and numerical
stability.

In general the approximate GCRD algorithm performs quite well for low
degree GCRDs and low noise. If the degree of the GCRD increases
relative to the degrees of $f$ and $g$ then $\Vhat + \Delta \Vhat$
becomes more unstructured and the perturbations become larger which
leads to unsuitable numeric GCRDs.  In the case of degree one or two
GCRDs we were still able to reconstruct meaningful answers even if the
noise was outside of our search radius because the matrix still
retained a Sylvester-like structure.

Despite the speculative nature of our algorithm we find that it works
reasonably well on the test data when we have a priori bound on the
noise and choose a search radius accordingly.  Approximating GCRDs of
higher order becomes more difficult with the unstructured approach as
the perturbations in the data will become
larger. If other assumptions are made about the structure of $f$ and $g$ and
the distribution of noise then it would be possible to obtain a better
approximate GCRD by exploiting this underlying structure.  



 \bibliographystyle{plainnat}

 \bibliography{approxore}  

\newcommand{\Gathen}{\relax}
\begin{thebibliography}{16}
\providecommand{\natexlab}[1]{#1}
\providecommand{\url}[1]{\texttt{#1}}
\expandafter\ifx\csname urlstyle\endcsname\relax
  \providecommand{\doi}[1]{doi: #1}\else
  \providecommand{\doi}{doi: \begingroup \urlstyle{rm}\Url}\fi

\bibitem[Beckermann and Labahn(1998)]{Beckerman97}
B.~Beckermann and G.~Labahn.
\newblock When are two numerical polynomials relatively prime?
\newblock \emph{Journal of Symbolic Computation}, 26:\penalty0 677--689, 1998.

\bibitem[Bini and Pan(1986)]{Bini86}
D.~Bini and V.~Pan.
\newblock Polynomial division and its computational complexity.
\newblock \emph{Journal of Complexity}, 2\penalty0 (3):\penalty0 179--203,
  1986.

\bibitem[Botting et~al.(2005)Botting, Giesbrecht, and May]{BGM05}
B.~Botting, M.~Giesbrecht, and J.P. May.
\newblock Using the {R}iemannian {SVD} for problems in approximate algebra.
\newblock In \emph{Proc. Workshop on Symbolic-Numeric Computation (SNC'05)},
  pages 209--219, 2005.

\bibitem[Bronstein and Petkov\v{s}ek(1994)]{BroPet94}
M.~Bronstein and M.~Petkov\v{s}ek.
\newblock On {O}re rings, linear operators and factorisation.
\newblock \emph{Programmirovanie}, 20:\penalty0 27--45, 1994.

\bibitem[Corless et~al.(1995)Corless, Gianni, Trager, and Watt]{CGTW95}
R.~M. Corless, P.~M. Gianni, B.~M. Trager, and S.~M. Watt.
\newblock The singular value decomposition for polynomial systems.
\newblock In \emph{Proc. International Symposium on Symbolic and Algebraic
  Computation (ISSAC'95)}, pages 189--205, 1995.

\bibitem[Corless et~al.(2004)Corless, Watt, and Zhi]{Corless04}
R.~M. Corless, S.~M. Watt, and L.~Zhi.
\newblock {QR} factoring to compute the {GCD} of univariate approximate
  polynomials.
\newblock \emph{IEEE Transactions on Signal Processing}, 52\penalty0 (12),
  2004.

\bibitem[\Gathen{von zur Gathen}(2003)]{MCA3}
J.~\Gathen{von zur Gathen}, J.~Gerhard.
\newblock \emph{Modern Computer Algebra}.
\newblock Cambridge University Press, New York, NY, USA, 2 edition, 2003.

\bibitem[Giesbrecht and Kim(2013)]{GieKim13}
M.~Giesbrecht and M.~Kim.
\newblock Computing the {H}ermite form of a matrix of {O}re polynomials.
\newblock \emph{Journal of Algebra}, 376:\penalty0 341--362, 2013.

\bibitem[Golub and van Loan(2013)]{GolLoa13}
G.~Golub and C.~van Loan.
\newblock \emph{Matrix Computations}.
\newblock Johns Hopkins University Press, Baltimore, USA, 4th edition, 2013.

\bibitem[Kaltofen et~al.(2005)Kaltofen, Yang, and Zhi]{KYZ05}
E.~Kaltofen, Z.~Yang, and L.~Zhi.
\newblock Structured low rank approximation of a {Sylvester} matrix.
\newblock In \emph{Proc. Workshop on Symbolic-Numeric Computation (SNC'05)},
  pages 69--83, 2005.

\bibitem[Kaltofen et~al.(2006)Kaltofen, Yang, and Zhi]{KalYan06}
E.~Kaltofen, Z.~Yang, and L.~Zhi.
\newblock Approximate greatest common divisors of several polynomials with
  linearly constrained coefficients and singular polynomials.
\newblock In \emph{Proc. International Symposium on Symbolic and Algebraic
  Computation (ISSAC'06)}, pages 169--176, 2006.

\bibitem[Karmarkar and Lakshman(1996)]{KarLak96}
N.~Karmarkar and Y.~N. Lakshman.
\newblock Approximate polynomial greatest common divisors and nearest singular
  polynomials.
\newblock In \emph{Proc. International Symposium on Symbolic and Algebraic
  Computation (ISSAC'96)}, pages 35--39, 1996.

\bibitem[Li and Nemes(1997)]{Li97}
Z.~Li and I.~Nemes.
\newblock A modular algorithm for computing greatest common right divisors of
  {O}re polynomials.
\newblock In \emph{Proc. International Symposium on Symbolic and Algebraic
  Computation (ISSAC'97)}, pages 282--289, 1997.

\bibitem[Ore(1933)]{Ore33}
O.~Ore.
\newblock Theory of non-commutative polynomials.
\newblock \emph{Annals of Mathematics. Second Series}, 34:\penalty0 480--508,
  1933.

\bibitem[Sasaki and Sasaki(1997)]{SasSas97}
T.~Sasaki and M.~Sasaki.
\newblock Polynomial remainder sequence and approximate {GCD}.
\newblock \emph{ACM SIGSAM Bulletin}, 31:\penalty0 4--10, 1997.

\bibitem[Zeng and Dayton(2004)]{ZenDay04}
Z.~Zeng and B.~H. Dayton.
\newblock The approximate {GCD} of inexact polynomials.
\newblock In \emph{Proc. International Symposium on Symbolic and Algebraic
  Computation (ISSAC'04)}, pages 320--327, 2004.

\end{thebibliography}
\end{document}